\documentclass[a4paper]{article}
\usepackage{amsmath,amsthm,amssymb,indentfirst}
\usepackage[english]{babel} 
\usepackage[big]{layaureo}

\theoremstyle{definition}
\newtheorem*{proposition}{Proposition}
\newtheorem{lemma}{Lemma}
\newtheorem*{theorem}{Theorem}

\title{Existence of Equilibrium Prices:\\ A Pedagogical Proof\thanks{I would like to thank Giulio Codognato and Marialaura Pesce for their helpful suggestions.}}
\author{Simone Tonin\thanks{Durham University Business School, Durham University, Durham, DH1 3LB, UK}}
\begin{document}
\maketitle
\begin{abstract}Under the same assumptions made by Mas-Colell et al. (1995), I develop a short, simple, and complete proof of existence of equilibrium prices based on excess demand functions. The result is obtained by applying the Brouwer fixed point theorem to a trimmed simplex which does not contain prices equal to zero. The mathematical techniques are based on some results obtained in Neuefeind (1980) and Geanakoplos (2003).
\end{abstract}

\section{Introduction}
This paper aims to provide a short, simple, and complete proof of existence of equilibrium prices under the same set of assumptions made in Mas-Colell et al. (1995). The existence result is derived by two lemmas which define, respectively, the non-empty compact convex set and the continuous function required by the Brouwer fixed point theorem. As usual, it is shown that the fixed point is an equilibrium price vector.\par
It would be pretentious to say that the proof is new as it follows closely some results of Neuefeind (1980) and Geanakoplos (2003). However, the approach developed here highlights clearly the central role of fixed point theorems in proving the existence of equilibrium prices. This can be interesting from a pedagogical and historical point of view as the two results are strongly linked. I further discuss this point in the last section.\par
Let me briefly outline the results of Neuefeind (1980) and Geanakoplos (2003). Neuefeind (1980) showed that, under a proper boundary behaviour of the excess demand function, it is possible to consider a trimmed simplex which does not contain prices equal to zero. In the first lemma I prove the same result by considering a different boundary condition. Geanakoplos (2003) added a perturbation in a correspondence considered by Debreu (1959) to make it single-valued and to apply the Brouwer fixed point theorem. The second lemma states this key result and I report its proof for completeness. Furthermore, by using the result of Neuefeind (1980), the boundary condition of the excess demand function is used in a different way from Geanakoplos (2003) (for a detailed analysis of the boundary condition see Ruscitti, 2012).

\section{Mathematical model}
Consider an exchange economy with $L$ commodities, $l=1,\dots,L$. Let $z:\mathbb{R}_{++}^L\rightarrow\mathbb{R}^L$ be the exchange economy's excess demand function defined over the set of positive price vectors. I make the same assumptions used in Proposition 17.C.1 of Mas-Colell et al. (1995):
\begin{enumerate}
\item[(A1)]$z(\cdot)$ is continuous;
\item[(A2)]$z(\cdot)$ is homogeneous of degree zero;
\item[(A3)]$p\cdot z(p)=0$ for all $p\in\mathbb{R}_{++}^L$ (Walras' law);
\item[(A4)]There is an $s>0$ such that $z_l(p)>-s$ for all $l$ and for all $p\in\mathbb{R}_{++}^L$.
\item[(A5)]If $p^n\rightarrow p$, where $p\neq 0$ and $p_l=0$ for some $l$,\footnote{The symbol $0$ denotes the origin in $\mathbb{R}^L$ as well as the real number zero.} then $$\mbox{max}\{z_1(p^n),\dots, z_L(p^n)\}\rightarrow+\infty.$$\end{enumerate}
Exchange economies with a positive aggregate endowment of each commodity and with consumers having continuous, strongly monotone, and strictly convex preferences satisfy Assumptions A1-A5 (see Proposition 17.B.2 in Mas-Colell et al., 1995).\par
Let $\Delta=\{p\in \mathbb{R}_+^L:\sum_l p_l=1\}$ be the unit simplex. Since the excess demand function is homogeneous of degree zero, its domain can be restricted to the interior of the unit simplex $\mbox{int}\Delta$. However, this is an open set and to apply the fixed point theorem it is required a closed set which contains only positive prices. For this reason, I define a trimmed simplex $\Delta^\epsilon=\{p\in \Delta:p_l\geq\epsilon, \mbox{ for all }l\}$, with $\epsilon\in(0,\frac{1}{L}]$. In the paper it is always assumed that $\epsilon$ lies in $(0,\frac{1}{L}]$. Finally, in the proof I will repeatedly use the price vector $\bar{q}=(\frac{1}{L},\dots,\frac{1}{L})$, at which all prices are equal, and the following straightforward result.
\begin{proposition} The sets $\Delta$ and $\Delta^\epsilon$ are non-empty, compact, and convex.\end{proposition}
I finally state the theorem of existence of equilibrium prices.
\begin{theorem}Under Assumptions A1-A5, there exists a $p^*\in\mathbb{R}_{++}^L$ such that $z(p^*)=0$.\end{theorem}
It is immediate to see that $z(p)=0$ corresponds to the system of L equations in L unknowns introduced by Walras (1874-7) and the solution $p^*$ is the equilibrium price vector at which all markets clear.

\section{Proof}
The first step consists in defining the non-empty, compact, and convex set required by the fixed point theorem. To this end, the result of the next lemma implies that if an equilibrium price vector $p^*$ exists, then there is an $\epsilon$ such that $p^*$ belongs to the interior of the set $\Delta^\epsilon$.
\begin{lemma}Let $Q=\{p\in \mbox{int}\Delta: \sum_l z_l(p)\leq 0\}$. Under Assumptions A1-A5, there exists an $\epsilon\in(0,\frac{1}{L}]$ such that $Q\subseteq \mbox{int}\Delta^\epsilon$.\end{lemma}
\begin{proof}First, the set $Q$ is non-empty as the price vector $\bar{q}$ belongs to $Q$ by Walras' law, i.e., $\frac{1}{L}\sum_l z_l(\bar{q})=0$. Next, I show, by contradiction, that there exists an $\epsilon\in(0,\frac{1}{L}]$ such that $Q\subseteq\mbox{int}\Delta^\epsilon$. Suppose that for all $\epsilon\in(0,\frac{1}{L}]$ there exists a price vector $p\in Q$ such that $p\notin\mbox{int}\Delta^\epsilon$. Consider a sequence of $\{\epsilon^n\}$ with $\epsilon^n=\frac{1}{nL}$. Then, there exists a sequence of price vectors $\{p^n\}$ such that $p^n\in Q$ and $p^n\notin \mbox{int}\Delta^{\epsilon^n}$ for all $n$. Since the sequence $\{p^n\}$ belongs to the compact set $\Delta$, there is a subsequence $p^{k_n}\rightarrow p$ with $p\in\Delta$. As $p^{k_n}\notin\mbox{int}\Delta^{\epsilon^{k_n}}$ for all $n$, it follows that, for all $n$, $p_l^{k_n}\leq\epsilon^{k_n}$ for some $l$. Then, I can conclude that $p_l=0$ for some $l$ because $\epsilon^{k_n}\rightarrow 0$. Hence, $\mbox{max}\{z_1(p^{k_n}),\dots,z_L(p^{k_n})\}\rightarrow\infty$ by A5. Moreover, since $z(\cdot)$ has a lower bound by A4, the inequality $\mbox{max}\{z_1(p^{k_n}),\dots, z_L(p^{k_n})\}-s(L-1)<\sum_{l}z_l(p^{k_n})$ holds for all $n$. As the left hand side converges to infinity by the previous result, there exists an $m$ such that $\mbox{max}\{z_1(p^{k_n}),\dots, z_L(p^{k_n})\}-s(L-1)>0$ for all $n>m$. But then, it follows that $\sum_l z_l(p^{k_n})>0$ for all $n>m$. As $p^{k_n}\in Q$ for all $n$, it also follows that $\sum z_l(p^{k_n})\leq 0$ for all $n$, a contradiction. Therefore, there exists an $\epsilon\in(0,\frac{1}{L}]$ such that $Q\subseteq \mbox{int}\Delta^\epsilon$.\end{proof}
The next step consists in defining the continuous function on $\Delta^\epsilon$ to itself required by the Brouwer fixed point theorem. A possible approach would be to consider the correspondence
$$\mu(p)=\underset{q\in \Delta^\epsilon}{\mbox{arg max}}\{q\cdot z(p)\},$$
which may be multivalued, and apply the Kakutani fixed point theorem. As noted in Debreu (1959), the use of this correspondence is motivated by the idea that if there is excess demand (supply) in a market the price should increase (decrease). Differently, I follow the approach of Geanakoplos (2003) who introduce a perturbation, $-\|q-p\|^2$, to make the correspondence above single-valued.\footnote{The symbol $\|x-y\|$ denotes the Euclidean distance between the vectors $x$ and $y$.} The next lemma deals with the continuous function $\phi(\cdot)$ on which I will apply the Brouwer fixed point theorem.
\begin{lemma}Let $\phi:\Delta^\epsilon\rightarrow \Delta^\epsilon$  be such that
$$
\phi(p)=\underset{q\in \Delta^\epsilon}{\mbox{arg\ max}}\left\{q\cdot z(p)-\|q-p\|^2\right\}.$$
Under Assumptions A1-A5, $\phi(\cdot)$ is a continuous function.\end{lemma}
\begin{proof}Define the function $g(q,p)=q\cdot z(p)-\|q-p\|^2$. Since I consider only prices belonging to $\Delta^\epsilon$, $g(\cdot,\cdot)$ is continuous as it is a sum of continuous function. Let $p$ be a constant. Then, $g(\cdot,p)$ has a maximum point on the non-empty compact set $\Delta^\epsilon$ by the Weierstrass Theorem. It is immediate to verify that the square of the Euclidean distance is strictly convex. In fact, its Hessian matrix is a diagonal matrix with positive entries and it is then positive definite. Then, $g(\cdot,p)$ is strictly concave because it is a sum of a linear function and a strictly concave function. But then, the maximum is unique and $\phi(\cdot)$ is a function. Finally, I prove that $\phi(\cdot)$ is continuous. Let $q^*=\phi(p)$ be the unique maximum point of $g(\cdot,p)$. Since $\Delta^\epsilon$ is compact, I can consider a sequence $p^n\rightarrow p$ and a corresponding subsequence $\{p^{k_n}\}$ such that $\phi(p^{k_n})\rightarrow r$ with $r\in\Delta^\epsilon$. By the definition of $\phi(\cdot)$, the following inequality holds $g(q^*,p^{k_n})\leq g(\phi(p^{k_n}),p^{k_n})$. As the function $g(\cdot,\cdot)$ is continuous, it follows that $g(q^*,p^{k_n})\rightarrow g(q^*,p)$ and $g(\phi(p^{k_n}),p^{k_n})\rightarrow g(r,p)$. Then, $g(q^*,p)\leq g(r,p)$ by the properties of sequences. Since the maximum point is unique, it follows that $r=q^*$. But then, if the subsequence $p^{k_n}\rightarrow p$, it follows that $\phi(p^{k_n})\rightarrow q^*=\phi(p)$. Hence, $\phi(\cdot)$ is continuous.\end{proof}
I finally prove the existence theorem
\begin{proof} By the results in Lemmas 1 and 2, I can apply the Brouwer fixed point theorem and then there exists a fixed point $p^*$ of the function $\phi(\cdot)$, i.e., $p^*=\phi(p^*)$. Next, I show that $p^*\in\mbox{int}\Delta^\epsilon$. Since $p^*$ is the unique maximum point of $g(\cdot,p^*)$ on $\Delta^\epsilon$ and $g(p^*,p^*)=0$ by Walras' law, it follows that 
$$g(\alpha p+(1-\alpha)p^*,p^*)=(\alpha p+(1-\alpha)p^*)\cdot z(p^*)-\|(\alpha p+(1-\alpha)p^*)-p^*\|^2<0,$$
for each $\alpha\in (0,1]$ and $p\in\Delta^\epsilon$ with $p\neq p^*$. By applying the Walras Law the previous inequality simplifies in
$$p\cdot z(p^*)<\alpha\|p-p^*\|^2,$$
for each $\alpha\in (0,1]$. Suppose now that there exists a $p\in \Delta^\epsilon$ such that $p\cdot z(p^*)>0$. Then, there exists a sufficiently small $\bar{\alpha}\in(0,1]$ such that $p\cdot z(p^*)>\bar{\alpha}\|p-p^*\|^2$, a contradiction. Hence, $p\cdot z(p^*)\leq 0$ for all $p\in\Delta^\epsilon$. Since $\bar{q}\in \Delta^\epsilon$, it follows that $\bar{q}\cdot z(p^*)=\frac{1}{L}\sum_lz_l(p^*)\leq 0$ which implies that $p^*\in Q$. But then, $p^*\in \mbox{int}\Delta^{\epsilon}$ by Lemma 1.  Furthermore, note that $p^*$ maximises $p\cdot z(p^*)$ on $\Delta^\epsilon$ because $p^*\cdot z(p^*)=0$, by Walras' law, and $p\cdot z(p^*)\leq 0$ for all $p\in\Delta^\epsilon$, by the previous result. Finally, I show that the fixed point $p^*$ is an equilibrium price vector, i.e., $z(p^*)=0$. I proceed by contradiction and I suppose that there exists a commodity $l$ such that $z_l(p^*)<0$. As $p^*$ maximises $p\cdot z(p^*)$ on $\Delta^\epsilon$, it follows that $p_l^*=\epsilon$. But $p^*\in \mbox{int}\Delta^{\epsilon}$, a contradiction. Hence, $z_l(p^*)\geq 0$ for all $l$. By the fact that $\sum_l z_l(p^*)\leq 0$, I can conclude that $z(p^*)=0$.\end{proof}

\section{Fixed point theorems and equilibrium prices}
In the literature there are alternative approaches to prove the existence of equilibrium prices which do not rely on fixed point theorems (see Greenberg (1977) and Quah (2008) among others). However they require stronger assumptions such as the weak gross substitutability or the weak axiom of revealed preference. This is due to the fact there is an equivalence between fixed point theorems and the existence of equilibrium prices under the classical assumptions as shown by Uzawa (1962). Debreu (1982) also pointed out that the proof of existence of equilibrium prices, under the classical assumptions, requires mathematical tools of the same power as fixed point theorems.\par
As the Brouwer fixed point theorem was published in 1911, it is not surprising that the problem of existence of equilibrium prices formulated by Walras (1874-7) was solved in a general way by McKenzie (1954) and Arrow and Debreu (1954) for the first time.\par
Under more restrictive assumptions some rigorous existence results were also obtained by A. Wald and J. von Neumann in the 1930s (for modern explanation of Wald's result see John, 1999). It is also worth mentioning that A. Wald wrote another paper on the existence problem in 1935 which unfortunately went lost. Duppe and Weintraub (2016) gave a detailed history of this lost proof and clarified that it applied a fixed-point theorem to show the existence of equilibrium prices in exchange economies.

\end{document}